%% file: mainpaper.tex
\documentclass[runningheads]{llncs}
\input{preamble}

\title{On How to Not Prove Faulty Controllers Safe in Differential Dynamic Logic%
\thanks{This work was supported by FFI, VINNOVA under grant number 2017-05519, \textit{Automatically Assessing Correctness of Autonomous Vehicles -- Auto-CAV}, and by the Wallenberg AI, Autonomous Systems and Software Program (WASP) funded by the Knut and Alice Wallenberg Foundation.}}
\titlerunning{On How to Not Prove Faulty Controllers Safe}
\author{%
Yuvaraj~Selvaraj\inst{1,3}\orcidID{0000-0003-2184-3069}
\and
Jonas~Krook\inst{1,3}\orcidID{0000-0002-9810-4697}
\and
Wolfgang~Ahrendt\inst{2}\orcidID{0000-0002-5671-2555}
\and
Martin~Fabian\inst{1}\orcidID{0000-0003-1287-9748}
}
\authorrunning{Y. Selvaraj et al.}
\institute{%
Department of Electrical Engineering, Chalmers University of Technology, Göteborg, Sweden\\ \email{\{yuvaraj,krookj,fabian\}@chalmers.se}
\and
Department of Computer Science and Engineering, Chalmers University of Technology, Göteborg, Sweden\\
\email{ahrendt@chalmers.se}
\and
Zenseact, Göteborg, Sweden\\
\email{\{firstname.lastname\}@zenseact.com}
}

\begin{document}

\maketitle
\begin{abstract}
Cyber-physical systems are often safety-critical and their correctness is crucial, as in the case of automated driving. Using formal mathematical methods is one way to guarantee correctness. Though these methods have shown their usefulness, care must be taken as modeling errors might result in proving a faulty controller safe, which is potentially catastrophic in practice. This paper deals with two such modeling errors in \emph{differential dynamic logic}. Differential dynamic logic is a formal specification and verification language for \emph{hybrid systems}, which are mathematical models of cyber-physical systems. The main contribution is to prove conditions that when fulfilled, these two modeling errors cannot cause a faulty controller to be proven safe. The problems are illustrated with a real world example of a safety controller for automated driving, and it is shown that the formulated conditions have the intended effect both for a faulty and a correct controller. It is also shown how the formulated conditions aid in finding a \emph{loop invariant} candidate to prove properties of hybrid systems with feedback loops. The results are proven using the interactive theorem prover KeYmaera~X.

\end{abstract}

\section{Introduction}\label{sec:intro}

\emph{Cyber-physical systems} (CPS) typically consist of a digital \emph{controller} that interacts with a physical dynamic system and are often employed to solve safety-critical tasks. For example, an automated driving system (ADS) has to control an autonomous vehicle (AV) to safely stop for stop signs, avoid collisions, etc. It is thus paramount that CPS work correctly with respect to their requirements. One way to ensure correctness of CPS is to use formal verification, which requires a formal model of the CPS. An increasingly popular family of models of CPS are \emph{hybrid} systems, which are mathematical models that combine discrete and continuous dynamics. 

To reason about the correctness of a CPS, hybrid systems can model different components of the CPS and their interactions, thus capturing the overall \emph{closed-loop} behavior. In general, hybrid systems that model real world CPS may involve three main components: a \emph{plant} model that describes the physical characteristics of the system, a controller model that describes the control software, and an \emph{environment} model that captures the behaviors of the surrounding world in which the controller operates, thereby defining the \emph{operational domain}. The goal for the controller is to choose control actions such that the requirements are fulfilled for \emph{all} possible behaviors of the hybrid system. 

Typically, the environment is modeled using nondeterminism to capture all possible behaviors. However, assumptions on the environment behavior are necessary to limit the operational domain and remove behaviors that are too hostile for any controller to act in a safe manner. For example, if obstacles are assumed to appear directly in front of an AV when driving, no controller can guarantee safety. While the assumptions in the formal models are necessary to make the verification tractable, there are subtle ways in which formal verification can provide less assurance than what is assumed~\cite{koopman2019credible}. In other words, as a result of the verification, the designer may conclude the controller to be safe in the entire assumed operational domain, whereas in reality some critical behaviors where the controller is actually at fault might be excluded from the verification. One possible cause for such a disparity between what is verified and what is assumed to be verified is the presence of modeling errors. In such cases, if a controller is verified to be safe, it leads to unsafe conclusions which might be catastrophic in practice.    
This paper deals with two such modeling errors by making them subject to interactive verification. In the first erroneous case, the environment assumptions and the controller actions interact in such a way that the environment behaves in a friendly way to adapt to 
the actions of a controller that exploits its friendliness. 
Then, it is possible that a faulty controller can be proven safe since the environment reacts to accommodate the bad control actions. 
An example of this is a faulty ADS controller that never brakes, together with an environment that reacts by always moving obstacles to allow the controller not to brake. 

In the second erroneous case, the assumptions about the environment and/or other CPS components remove all behaviors in which any action by the controller is needed. In this case, the assumptions over-constrain the allowed behaviors. For example, if the assumptions restrict the behavior of the AV to an extent that only braking is possible, then a faulty ADS controller can be proven safe because nothing is proven about the properties of the controller. In the worst-case, the assumptions remove all \emph{possible} behaviors, thereby making the requirement vacuously true. 

In both cases, a faulty controller can be proven safe with respect to the requirements for the wrong reasons, i.e., unintended modeling errors, thus resulting in potentially catastrophic operation of the CPS in practice. Modeling errors are in general hard to address because every model is an abstraction and there exists no ubiquitous notion of what a \emph{correct model} means. Therefore, a systematic way to identify and avoid modeling errors is highly desirable as it reduces the risk of unsound conclusions when a model is formally proven safe with respect to the requirements. Typically, the requirements specify (un)desired behavior of the closed-loop system within the operational domain and are expressed in some logical formalism to apply formal verification. \emph{Differential dynamic logic}~(\dl)~\cite{platzer2012logics,platzer2018logical} is a specification and verification language that can be used to formally describe and verify hybrid systems. The interactive theorem prover \keyx~\cite{fulton2015keymaera} implements a sound proof calculus~\cite{platzer2012logics,platzer2018logical} for~\dl and can thus mathematically prove that the models fulfill their specified requirements. 

The main contributions of this paper, Theorem~\ref{thm:cheating:controller} and Theorem~\ref{thm:passenger:controller}, formulate and prove conditions that when fulfilled, ensure the model cannot be proven safe if it is susceptible to the above modeling errors. Essentially, a loop invariant is used not only to reason about the model inductively but also to ensure that the interaction between the controller and the other components in the model is as intended; the two theorems provide conditions on the relation between the assumptions and the loop invariant. Furthermore, these conditions give hints as to when a suggested loop invariant for the model is sufficiently strong to avoid modeling errors. The problems are illustrated with a running example of an automated driving controller that shows that they can appear in real models. It is then proven that the formulated conditions have the intended effect. Finally, it is shown by example that the method captures the problematic cases and also increases confidence in a correct model free from the considered modeling errors.

\section{Preliminaries}

The logic \dl uses \emph{hybrid programs}~(HP) to model hybrid systems. An HP $\upalpha$ is defined by the following grammar, where $\upalpha,\,\upbeta$ are HPs, $x$ is a variable, $e$ is a term\footnote{Terms are polynomials with rational coefficients defined by $e,\,\Tilde{\!e}\;\Coloneqq x\;|\;c\in\rational\;|\;e+\,\Tilde{\!e}\;|\;e\cdot\,\Tilde{\!e}$.}, and $P$ and $Q$ are formulas of first-order logic of real arithmetic (FOL)\footnote{First-order logic formulas of real arithmetic are defined by $P,\,Q\Coloneqq e\ge\,\Tilde{\!e}\;|\;e=\,\Tilde{\!e}\;|\;\neg P\;|\;P \land Q\;|\;P \lor Q\;|\;P\to Q\;|\;P\leftrightarrow Q\;|\;\forall x P\;|\;\exists x P$ }:
\begin{equation*}\label{eqn:HPgrammar}
\upalpha\,\Coloneqq x \coloneqq e \; | \; x\coloneqq * \; | \; ?P \; | \; x'=f(x)\,\&\,Q \; | \; \upalpha\,\cup\,\upbeta \; | \; \upalpha;\,\upbeta \; | \; \upalpha^* 
\end{equation*}

Each HP $\upalpha$ is semantically interpreted as a reachability relation $\hpsb{\upalpha} \subseteq \stateset \times \stateset$, where $\stateset$ is the set of all states. If $\varset$ is the set of all variables, a state $\upomega \in \stateset$ is defined as a mapping from $\varset$ to $\real$, i.e., $\upomega \colon \varset \to \real$. The notation $(\upomega, \upnu) \in \hpsb{\upalpha}$ denotes that final state $\upnu$ is reachable from initial state $\upomega$ by executing the HP $\upalpha$. $\upomega \hpsb{e}$ denotes the value of term $e$ in state $\upomega$, and for $x \in \varset$, $\upomega(x) \in \real$ denotes the real value that variable $x$ holds in state $\upomega$. Given a state~$\upomega_1$, a state~$\upomega_2$ can be obtained by assigning the terms $\{ e_1, \dotsc, e_n \}$ to the variables $y=\{y_1, \dotsc, y_n \} \subseteq \varset$ and letting the remaining variables in~$\varset$ be as in~$\upomega_1$, that is, $\upomega_2(y_i) = \upomega_1\hpsb{e_i}$ for $1 \le n$ and $\upomega_2(v) = \upomega_1(v)$ for all $v\in\varset\setminus y$. Let $\upomega_2 = \upomega_1(y_1\coloneqq e_1, \dotsc, y_n \coloneqq e_n)$ be a shorthand for this assignment. 
For a FOL formula $P$, let $\dlsb{P} \subseteq \stateset$ be the set of all states where $P$ is true, thus 
$\upomega \in \dlsb{P}$ denotes that $P$ is true in state $\upomega$.
If~$P$ is parameterized by $y_1, \dotsc, y_n$, then $\upomega \in \dlsb{P}$ means that $\upomega \in \dlsb{P(\upomega(y_1),\dotsc,\upomega(y_n))}$. A summary of the program statements in HP and their transition semantics~\cite{platzer2018logical} is given in Table~\ref{tab:hp}.

\begin{table}[!ht]\small
    \caption{Semantics of HPs~\cite{platzer2018logical}. $P,Q$ are first-order formulas, $\upalpha,\upbeta$ are HPs.}
    \label{tab:hp}
    \centering
    \begin{tabularx}{\linewidth}{@{}p{6.5em} @{\quad=\,} X@{}}
         \toprule
         \multicolumn{1}{l}{Statement} & \multicolumn{1}{l}{Semantics} \\
         \midrule
         $\hpsb{x\coloneqq e}$ & $\Bigl\{\, (\upomega,\upnu) : \upnu = \upomega(x\coloneqq e) \,\Bigr\}$ \\
         $\hpsb{x\coloneqq *}$ & $\Bigl\{\, (\upomega,\upnu) : c \in \real \textrm{ and } \upnu = \upomega(x \coloneqq c)\,\Bigr\}$ \\
         $\hpsb{?P}$ & $\Bigl\{\, (\upomega, \upomega) : \upomega \in \dlsb{P}\, \Bigr\}$\\
         $\hpsb{x'=f(x)\,\&\,Q}$ & $\Bigl\{\, (\upomega, \upnu) : \upphi(0)=\upomega(x'\coloneqq f(x))\ \text{and}\ \upphi(r) = \upnu\ \text{for a solution}\ \upphi : [0,r] \to \stateset\ \text{of any duration}\ r\ \text{satisfying}\ \upphi \models x'=f(x) \land Q \,\Bigr\}$\\
         $\hpsb{\upalpha \cup \upbeta}$  & $ \hpsb{\upalpha} \cup \hpsb{\upbeta}$ \\
         $\hpsb{\upalpha;\,\upbeta}$ & $\hpsb{\upalpha} \circ \hpsb{\upbeta} = \Bigl\{\, (\upomega, \upnu) : (\upomega, \upmu) \in \hpsb{\upalpha}, (\upmu, \upnu) \in \hpsb{\upbeta} \,\Bigr\}$\\
         $\hpsb{\upalpha^*}$ & $\hpsb{\upalpha}^* = \bigcup\limits_{n\in\naturalnum_0} \hpsb{\upalpha^n}\ \text{with}
         \ \upalpha^{0} \equiv\, ?true\ 
         \text{and}\ \upalpha^{n+1} \equiv \upalpha^{n};\upalpha
         .$\\
         \bottomrule
    \end{tabularx}\vspace{-0.4cm}
\end{table}

The sequential composition $\upalpha;\,\upbeta$ expresses that $\upbeta$ starts executing after $\upalpha$ has finished. The \emph{nondeterministic choice} operation expresses that the HP $\upalpha \cup \upbeta$ can nondeterministically choose to follow either $\upalpha$ or $\upbeta$. The \emph{test} action $?P$ has no effect in a state where $P$ is true, i.e., the final state $\upomega$ is same as initial state $\upomega$. If however $P$ is false when $?P$ is executed, then the current \run\ of the HP \emph{aborts}, meaning that no transition is possible and the entire current \run\ is removed from the set of possible behaviors of the HP. The \emph{nondeterministic repetition} $\upalpha^*$ expresses that $\upalpha$ repeats $n$ times for any $n \in \naturalnum_0$. Furthermore, test actions can be combined with sequential composition and the choice operation to define \emph{if-statements} as: 
\begin{equation}\label{eq:ifelse}
    \opif\,(P)\; \opthen\; \upalpha\; \opfi \equiv (?P;\,\upalpha) \cup (?\neg P)
\end{equation}

HPs model continuous dynamics as $ x'=f(x)\,\&\,Q$, which describes the \emph{continuous evolution} of $x$ along the differential equation system $x'=f(x)$ for an arbitrary duration (including zero) within the \emph{evolution domain constraint} $Q$. The evolution domain constraint applies bounds on the continuous dynamics and are first-order formulas that restrict the continuous evolution within that bound. $x'$ denotes the time derivative of $x$, where $x$ is a vector of variables and $f(x)$ is a vector of terms of the same dimension.

The formulas of \dl include formulas of first-order logic of real arithmetic and the modal operators $[\upalpha]$ and $\langle\upalpha\rangle$ for any HP $\upalpha$~\cite{platzer2012logics,platzer2018logical}. A formula~$\uptheta$ of \dl is defined by the following grammar ($\upphi,\,\uppsi$ are \dl formulas, $e,\,\Tilde{\!e}$ are terms, $x$ is a variable, $\upalpha$ is an HP):
\begin{equation}\label{eq:dlgrammar}
    \uptheta \Coloneqq e=\,\Tilde{\!e}\;|\;e\ge \,\Tilde{\!e}\;|\;\neg \upphi\;|\;\upphi \land \uppsi\;|\;\forall x\,\upphi\;|\;[\upalpha]\,\upphi
\end{equation}
The \dl formula $[\upalpha]\,\upphi$ expresses that \emph{all} non-aborting \run s of HP $\upalpha$ (i.e., the \run s where all test actions are successful) end in a state in which the \dl formula $\upphi$ is true. The formal semantics are defined by $\dlsb{[\upalpha]\,\upphi} = \{ \upomega \in \stateset : \forall \upnu \in \stateset \ldotp (\upomega, \upnu) \in \hpsb{\upalpha} \to \upnu \in \dlsb{\upphi} \}$. The \dl formula $\langle\upalpha\rangle\,\upphi$ means that there exists \emph{some} non-aborting \run\ leading to a state where $\upphi$ is true.  $\langle\upalpha\rangle\,\upphi$ is the dual to $[\upalpha]\,\upphi$, defined as $\langle\upalpha\rangle\,\upphi \equiv \neg[\upalpha]\neg\upphi$. Similarly, $>,\le,<,\lor,\to,\leftrightarrow,\exists x$ are defined using combinations of the operators in~\eqref{eq:dlgrammar}. A \dl formula $\uptheta$ is \emph{valid}, denoted $\models \uptheta$ if $\dlsb{\uptheta} = \stateset$.

The logic \dl and the interactive theorem prover \keyx support the specification and verification of hybrid systems. The \dl formula $(\assume) \to [\,\upalpha\,]\,(\guarantee)$ can be used to specify the correctness of an HP $\upalpha$ with respect to the requirement \guarantee. It expresses that, if the initial conditions described by the formula $\assume$ are true, then all (non-aborting) \run s of $\upalpha$ only lead to states where formula $\guarantee$ is true. \keyx takes such a \dl formula as input and successively decomposes it into several sub-goals according to the sound proof rules of \dl to prove the formula~\cite{platzer2012logics,platzer2018logical}. 

Often, modeling CPS as HPs involves execution of a controller together with a plant in a loop described by the nondeterministic repetition $\upalpha^*$. To prove properties of loops, like the property $(\assume) \to [\,\upalpha^*\,]\,(\guarantee)$, \keyx uses \emph{loop invariants}, provided by the user, to inductively reason about all (non-aborting) executions. Given a loop invariant (candidate) $\upzeta$, applying the loop invariant rule to the above formula would make the proof branch into the following three cases:
\begin{enumerate}[label=\texttt{loop} (\roman*) , wide=0pt]
    \item: $(\assume) \to \upzeta$, i.e., the initial state satisfies the invariant,
    \item: \label{it:invariant:loop} $\upzeta \to [\,\upalpha\,]\,\upzeta$, i.e., invariant remains true after one iteration of $\upalpha$ from any state where the invariant was true,
    \item: $\upzeta \to (\guarantee)$, i.e., the invariant implies the requirement.
\end{enumerate}

\section{Problem Scope}\label{sec:scope}
The scope of CPS considered in this paper are hybrid systems with closed-loop feedback control as described by Model~\ref{alg:general}. The \dl formula~\eqref{eq:general:dl} models the CPS as a HP that repeatedly executes in a loop and expresses the requirement on the CPS by the formula \guarantee. The HP in~\eqref{eq:general:dl} is composed of four different components, each of which is an HP and assigns four variables: the dynamic state~$s$ which evolves continuously, the control actions~$a$, the environment actions~$e$, and the time progress~$\tau$. Though the variables in Model~\ref{alg:general} are scalars, they can in general be vectors of any dimension. 

\begin{figure}[!t]
 \removelatexerror
  \begin{HP}[H]
    \DontPrintSemicolon
    \setstretch{0.5}
    \caption{The general model considered}
    \label{alg:general}
    \SetAlgoLined
    
    \begin{flalign} \label{eq:general:dl}
        (\assume) \to [(\env;\,\sys;\,\ctrl;\,\plant)^*]\,(\guarantee) &&
    \end{flalign}
    \begin{flalign}
        \env \triangleq &\:e\coloneqq*;\,?\,P(s,e,a) & \label{eq:general:env}
    \end{flalign}
    \begin{flalign}\label{eq:general:sys}
        \sys \triangleq a\coloneqq *;\,?\,Q(s, e, a) &&
    \end{flalign}
    \begin{flalign}\label{eq:general:ctrl}
             \ctrl \triangleq \opif \, \lnot\safe(s,e,a)\, \opthen \, a\coloneqq *;\,?\,C(s,e,a) \, \opfi &&
    \end{flalign} 
    \begin{flalign}\label{eq:general:plant}
        \plant \triangleq \tau\coloneqq 0;\, s' = f(s),\,\tau' = 1\; \& \;  F(s,e,a,\tau) \land \tau \le \Tau  &&
    \end{flalign}   
  \end{HP}\vspace{-0.4cm}
\end{figure}

The environment~(\env) in~\eqref{eq:general:env} describes the environment behavior using a nondeterministic assignment followed by a test. The environment action~$e$ is nondeterministically assigned a real value which is then checked by the subsequent test for adherence to the environment assumptions~$P$, which define the operational domain. The auxiliary system~(\sys) describes the internal digital system that the controller interacts with, in addition to the environment. Similarly to \env, \sys~\eqref{eq:general:sys} nondeterministically assigns a real value to the control action~$a$ followed by a subsequent test which checks whether the internal assumptions~$Q$ hold. These internal assumptions typically describe conditions that stem from the design of the CPS such as physical limits on the system actuators.

The controller's~(\ctrl) task is to ensure that the requirement \guarantee is fulfilled and is modeled as an if-statement as seen in~\eqref{eq:general:ctrl}. First, the control action~$a$ set by \sys is tested with \safe. If the test is not \safe, then \ctrl overrides the control action~$a$ by the control law~$C$, and finally it passes on the control action to the \plant~\eqref{eq:general:plant}, which models the physical part of the system. It is described as an ordinary differential equation. However, the sampling time of \ctrl is bounded, so the evolution of \plant must stop before the sampling time~$T$~\cite{selvaraj2022formal}.

In the most abstract setting, the parameterized FOL formulas in Model~\ref{alg:general} are treated as uninterpreted predicates, which could be replaced by any concrete hybrid model with specific formulas and HPs, as long as the assignment of values to variables follows the flow of Model~\ref{alg:general}. Hence, the conclusions drawn from Model~\ref{alg:general} can be applied and used for a wide variety of hybrid systems.

\subsection*{Running Example: Automated Driving Controller}
To illustrate the problems and solutions, this paper considers an example of an in-lane automated driving feature for an AV, the \emph{\ego}. \fig~\ref{fig:funcarch} shows a simplified architecture of the automated driving feature, which can be modeled as a HP of the general form in Model~\ref{alg:general}. The safety requirement is for the \ego to safely stop for stationary obstacles that have entered its path.

The \emph{perception} senses the world around the \ego and corresponds to the \env in Model~\ref{alg:general}. The \env models the perception algorithms that communicate the obstacle position~$x_c$ to the controller and thus the \env assumptions describe the dynamics of the obstacles appearing in the \ego{}'s path. The \emph{nominal controller}, described by \sys, represents any algorithm solving the nominal driving task subjected to different constraints (e.g. comfort) and requests a nominal acceleration. Thus, \sys of the form in~\eqref{eq:general:sys} allows to keep the model parametric to arbitrary nominal controller implementations while being regarded as a black box. The \sys assumptions therefore capture design conditions on the nominal controller such as always requesting an acceleration within certain bounds. 

The \emph{safety controller} described by \ctrl ensures that only safe control actions, i.e., acceleration commands~$a$, are communicated to the actuators. It evaluates the nominal acceleration and overrides it with a safe acceleration if needed, thereby satisfying the safety requirement. Thus, the verification of the safety requirement can be limited to verifying the decision logic in one component, the safety controller.

\tikzset{
     block/.style={rectangle, draw,
                   text width=5em,
                   text centered, rounded corners, minimum height=3em, on grid},
     arrow/.style={-{Stealth[]},thick},
     no arrow/.style={-,every loop/.append style={-},thick},
     smallblock/.style={draw, rectangle,text width=4em, text centered, rounded corners, minimum height=2em}
     }

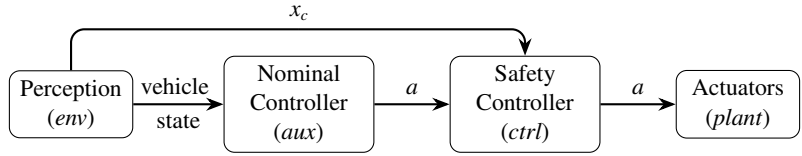
\begin{figure}[!t]
    \centering
    \begin{tikzpicture}[auto, font=\small]
    
    \node [smallblock] (sense) {Perception \\ (\env)};
    \node [block, right=1.2cm of sense.east] (nominal) {Nominal Controller \\ (\sys)};
    \node [block, right=of nominal.east] (safety) {Safety Controller \\ (\ctrl)}; 
    \node [smallblock, right=of safety.east] (vehicle) {Actuators \\ (\plant)};
    \node [coordinate, above=1em of nominal, label={$x_c$}] (t) {};
    
    \draw [arrow] (sense) --  node[name=v,above] {vehicle\;}
    node[name=s,below] {state}
    (nominal);
    \draw [arrow] (safety) -- node[name=as] {$a$} (vehicle);
    \draw [arrow] (nominal) -- node[name=an] {$a$}(safety); 
    \draw [no arrow, rounded corners=5pt] (sense) |- (t);
    \draw [arrow, rounded corners=5pt] (t) -| (safety);
    
    \end{tikzpicture}
    \caption{Architecture of the automated driving feature.}
    \label{fig:funcarch}\vspace{-0.4cm}
\end{figure}

The \plant is a dynamic model of the \ego. It is modeled as a double integrator with position~$x$ and velocity~$v$ of the \ego as the dynamic states and the acceleration~$a$ as control input, as seen in~\eqref{eq:m1:plant} of Model~\ref{alg:m1}. The \ego is not allowed to drive backwards, so $v$ must be non-negative through the entire evolution. In other words, the evolution would stop before $v$ gets negative.

In the next section, the general \dl formula in~\eqref{eq:general:dl} is refined with concrete descriptions of \env,~\sys, and~\ctrl to illustrate the modeling errors where a faulty controller can be proven safe. However, \assume, \plant, and \guarantee remain unchanged in the subsequent models and are shown in Model~\ref{alg:m1}. The initial condition \assume~\eqref{eq:m1:assume} specifies that the \ego starts stationary ($v=0$) at an arbitrary position $x$ before the position $x_c$ of an obstacle. It also sets up assumptions on the \emph{constant} parameters such as the minimum safety and nominal acceleration $\aMinS$ and $\aMinN$, and maximum nominal acceleration, $\aMaxN$, and that the sampling time~$\Tau$ is positive. These constant parameters do not change value during the execution of the HP $[(\env;\,\sys;\,\ctrl;\,\plant)^*]$, and therefore the assumptions on the constant parameters remain true in all contexts. The requirement that the \ego must stop before stationary obstacles is expressed by the post condition \guarantee~\eqref{eq:m1:guarantee}, which says that the obstacle's position may not be exceeded. 

\begin{figure}[!t]
 \removelatexerror
  \begin{HP}[H]
    \DontPrintSemicolon
    \setstretch{0.5}
    \caption{Example hybrid system}
    \label{alg:m1}
    \SetAlgoLined
    \noindent
    \begin{flalign}\label{eq:m1:assume}
        \assume \triangleq &\: v = 0 \land x \le x_c \land \aMinS > 0 \land \aMaxN > 0  &\nonumber\\ 
                          &  \land \aMinN > 0 \land \aMinS > \aMinN \land \Tau>0
    \end{flalign}
    \begin{flalign}\label{eq:m1:guarantee}
        \guarantee \triangleq (x \le x_c) &&
    \end{flalign}
    \begin{flalign}\label{eq:m1:plant}
        \plant \triangleq \tau\coloneqq 0;\, x' = v,v' = a,\tau' = 1\; \& \;  v \ge 0 \land \tau \le \Tau  &&
    \end{flalign}   
  \end{HP}\vspace{-0.4cm}
\end{figure}

\section{Discover Modeling Errors}

This section presents two erroneous models to illustrate how a faulty \ctrl can be proven safe with respect to~\guarantee. In the first case, shown in Model~\ref{alg:m2}, improper interaction between \env and \ctrl results in \env adapting to faulty \ctrl actions. 
Such an erroneous model can be proven safe since the loop invariant~$\upzeta$
is not strong enough to prevent improper interactions.
Theorem~\ref{thm:cheating:controller} gives conditions to strengthen~$\upzeta$ to avoid such issues. 
In the second erroneous case, Model~\ref{alg:m3}, the error arises due to over-constrained \env and \sys assumptions that discard executions where \ctrl is at fault.  Theorem~\ref{thm:passenger:controller} presents conditions to identify and avoid errors due to such over-constrained assumptions. 

\subsection{Exploiting Controller}\label{sec:ctrl:cheats}

Consider Model~\ref{alg:m2} where the assumptions on \env and \sys are given by~\eqref{eq:m2:env} and~\eqref{eq:m2:sys} respectively. \env assigns~$x_c$ such that it is possible to brake and stop before the position of the obstacle. This is necessary since if an obstacle appears immediately in front of the moving \ego it is physically impossible for any controller to safely stop the vehicle. \sys is a black box, but it is known that the nominal acceleration request~$a$ is bounded. The \ctrl test \safe~\eqref{eq:m2:safe} checks whether maximal acceleration for a time period of~$\Tau$ leads to a violation of the requirement, and if it does, the controller action~$C$~\eqref{eq:m2:ctrllaw} sets the deceleration to its maximum. This maximum deceleration is a symbolic value, parameterized over the other model variables.

\begin{figure}[!t]
 \removelatexerror
  \begin{HP}[H]
    \DontPrintSemicolon
    \setstretch{0.5}
    \caption{\ctrl is exploiting}
    \label{alg:m2}
    \SetAlgoLined
    \noindent
    \begin{flalign}
        \env \triangleq &\:x_c\coloneqq*;\,?\,\left(x_c-x \ge \frac{v^2}{2\aMinN} \right) & \label{eq:m2:env}
    \end{flalign}
    \begin{flalign}\label{eq:m2:sys}
        \sys \triangleq a\coloneqq *;\,?\,(-\aMinN \le a \le \aMaxN) &&
    \end{flalign}
    \begin{flalign}\label{eq:m2:ctrl}
             \ctrl \triangleq \opif \, \lnot\safe(x,v,x_c,a) \, \opthen\, a\coloneqq *;\,?\,C(x,v,x_c,a)\, \opfi &&
    \end{flalign} 
    \begin{flalign}\label{eq:m2:safe}
        \safe(x,v,x_c,a) \triangleq \left( x_c - x \ge v\Tau + \dfrac{\aMaxN\,\Tau^2}{2} \right) &&
    \end{flalign}
    \begin{flalign}\label{eq:m2:ctrllaw}
        C(x,v,x_c,a) \triangleq a = -\aMinS &&
    \end{flalign}
  \end{HP}\vspace{-0.4cm}
\end{figure}

Denote by~$\uptheta$  the \dl formula~\eqref{eq:general:dl} together with the definitions of Model~\ref{alg:m1} and Model~\ref{alg:m2}. $\uptheta$ is proved~\cite{selvaraj2022:models} with the loop invariant $\upzeta_1 \equiv x \le x_c$. Though the goal is to find a proof that $\uptheta$ is valid, and thereby establish that \ctrl is safe with respect to \guarantee, it is in this case incorrect to draw that conclusion from the proof, as will now be shown. 

The \env assumption~\eqref{eq:m2:env} discards executions where the distance between the obstacle position~$x_c$ and the \ego position~$x$ is less than the 
minimum possible
braking distance of the \ego. This  assumption is reasonable as it only discards situations where it is physically impossible for \ctrl to safely stop the vehicle. Still, infinitely many \env behaviors are possible since $x_c$ is nondeterministically assigned any value that fulfills the assumption. Among other behaviors, this allows $x_c$ to remain constant, as would be the case for stationary obstacles.
However, due to improper interaction between \env and a faulty \ctrl, \env can be forced by \ctrl to not have $x_c$ constant.

Consider a state $\upomega_0 \in \dlsb{\upzeta_1}$, illustrated in \fig~\ref{fig:cheating:controller:step}, such that
\begin{align*}
    \upomega_0 \ssb{x} &= 0 &
    \upomega_0 \ssb{x_c} &= 1 &
    \upomega_0 \ssb{T} &= 1 & \\
    \upomega_0 \ssb{v} &= 0 &
    \upomega_0 \ssb{a} &= 1.8 &
    \upomega_0 \ssb{\aMaxN} &= 2 &
    \upomega_0 \ssb{\aMinN} = 3 \ .    
\end{align*}

The \ego is currently at $(x,v) = (0,0)$ as shown by the black circle. The hatched area represents all the points in the $xv$-plane from which it is possible to stop before the obstacle position, $x_c$, at the dashed vertical line. It holds that $(\upomega_0, \upomega_0)\in\hpsb{\env}$ since $x_c - x = 1 \ge 0^2/(2\times 3)=v^2/(2\aMinN)$, so the assumptions on \env allow $x_c = 1$. This can also be seen in the figure since the black circle is within the hatched area. The arrow labeled~$a$ in \fig~\ref{fig:cheating:controller:step} represents the acceleration request by \sys, and if \plant evolves for 1~second with~$a$ as input, the \ego ends up at the white circle. As~$a$ is within the bounds of \sys, it holds that $(\upomega_0, \upomega_0)\in\hpsb{\sys}$. The controller \ctrl is \safe with this choice since~$x_c$ is not passed if maximum acceleration~$\aMaxN$ is input to \plant, as illustrated by the gray circle in the figure. Formally, $x_c - x = 1 \ge 0\times 1 + 2\times 1^2/2 = vT + \aMaxN T^2/2$ and therefore it holds by~\eqref{eq:m2:safe} that $\upomega_0 \in \dlsb{\safe(x, v, a_n, a)}$. Thus, $(\upomega_0, \upomega_0) \in \hpsb{\ctrl}$. Let $\upomega_1 = \upomega_0(x \coloneqq 0.9, v \coloneqq 1.8)$. Now it holds that $(\upomega_0, \upomega_1) \in \hpsb{\plant}$, i.e., starting at $x=0$ and $v=0$, with $a=1.8$ as input, \plant evolves to $x=0.9$ and $v=1.8$ in 1~second.
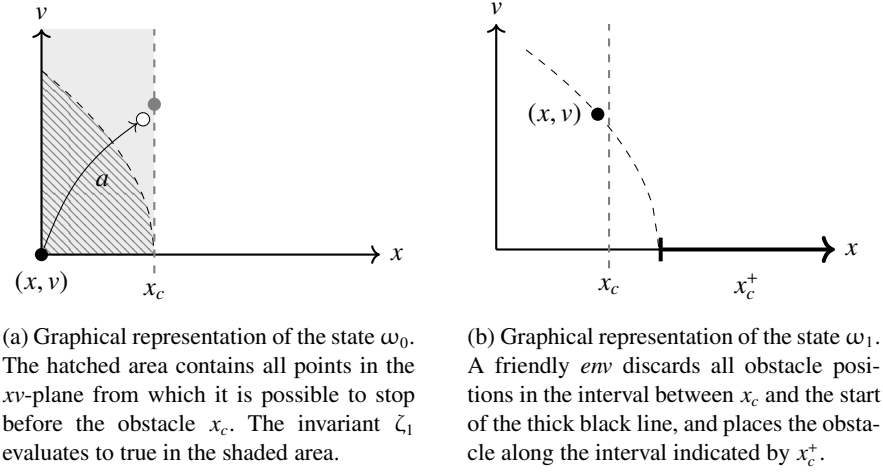
\begin{figure}[tbp]
    \centering
    \subcaptionbox{Graphical representation of the state~$\upomega_0$. The hatched area contains all points in the $xv$-plane from which it is possible to stop before the obstacle~$x_c$. The invariant~$\upzeta_1$ evaluates to true in the shaded area.\label{fig:cheating:controller:step}}[.45\linewidth]{
        \begin{tikzpicture}
        \fill [fill=gray!15] (1.5,3) -- (1.5,0) -- (0,0) -- (0,3) -- cycle;
        \draw [xscale=1.5, pattern=north west lines, pattern color=gray, dashed, draw=black] plot [domain=0:1, smooth, variable=\x] ({\x}, {sqrt(2*3*(1-\x))}) -- (1,0) -- (0,0) -- (0,2.449) -- cycle;

        \draw [->, thick, xscale=1.5] (0,0) -- (3,0) node[right] {$x$};
        \draw [->, thick] (0,0) -- (0,3) node[above] {$v$};
        \draw [dashed, gray, thick, xscale=1.5] (1,-0.25) node [below, black] {$x_c$} -- (1,3);
        \node [circle, minimum size=5pt, inner sep=0pt, fill=black, label=below:{$(x, v)$}] (xv) at (0,0) {};
        \draw [xscale=1.5] (0.9,1.8) node (x1) [circle, draw=black,fill=white, inner sep=0, minimum size=5pt] {};
        \draw [->, out=75, in=-135, looseness=1, xscale=1.5] (xv) edge node [label=right:{$a$}] {} (x1);
        \draw [xscale=1.5] (1,2) node (xM) [circle, fill=gray, inner sep=0, minimum size=5pt] {};
        \end{tikzpicture}
    }\hspace{0.05\linewidth}
    \subcaptionbox{Graphical representation of the state~$\upomega_1$. A friendly \env discards all obstacle positions in the interval between $x_c$ and the start of the thick black line, and places the obstacle along the interval indicated by $x_c^+$.\label{fig:cheating:controller:obstacle}}[.45\linewidth]{
        \begin{tikzpicture}
        \draw [-, thick, xscale=1.5] (0,0) -- (3,0) node[right] {$x$};
        \draw [->, thick] (0,0) -- (0,3) node[above] {$v$};
        \draw [dashed, gray, thick, xscale=1.5] (1,-0.25) node [below, black] {$x_c$} -- (1,3);
        \draw [xscale=1.5] (0.9,1.8) node (x1) [circle, fill=black, inner sep=0, minimum size=5pt, label=left:{$(x, v)$}] {};
        \draw [domain=0.265:1.44, smooth, variable=\x, dashed, xscale=1.5] plot ({\x}, {sqrt(2*3*(1.44-\x))}) node [below] {};
        \draw [|->, line width=1.5pt, xscale=1.5] (1.44,0) -- node [below, yshift=-5pt] {$x^+_c$} (3,0);
        \end{tikzpicture}
    }
    \caption{The controller chooses an action such that the \plant evolves to a state where $x\le x_c$. In the next loop iteration, \env moves~$x_c$ to adapt to the controller's action.}
    \label{fig:cheating:controller}\vspace{-0.4cm}
\end{figure}

After \plant has evolved and the system has transited to $\upomega_1$, the \ego is now at the black circle in \fig~\ref{fig:cheating:controller:obstacle}. It is clear that $\upomega_1 \in \dlsb{\upzeta_1}$ as $x\le x_c$. The intersection of the dashed curve with the $x$-axis in \fig~\ref{fig:cheating:controller:obstacle} represents the lower bound for~$x_c$ to satisfy~\eqref{eq:m2:env} in the state~$\upomega_1$. Therefore, in the next iteration, $x_c$ can only be positioned somewhere along the interval indicated by the thick black line in \fig~\ref{fig:cheating:controller:obstacle} and all other values are discarded by~\eqref{eq:m2:env}. Semantically, as $x_c -x = 0.1 < 2^2/(2\times3) = v^2/(2\aMinN)$, it follows that $(\upomega_1, \upomega_1) \not\in \hpsb{\env}$ so $x_c$ cannot be kept constant between iterations.

To summarize, it holds that $\upomega_0 \in \dlsb{\upzeta_1}$, $(\upomega_0, \upomega_1) \in \hpsb{\env;\sys;\ctrl;\plant}$, and $\upomega_1 \in \dlsb{\upzeta_1}$. The acceleration requested by \sys is \safe{}'d by \ctrl in~$\upomega_0$ because the worst-case acceleration~$\aMaxN$ in~$\upomega_0$ leads to a state that fulfills~$\upzeta_1$, and therefore also fulfills $\guarantee$. Since there exists no control action allowed by the system dynamics in the assumed operational domain that can fulfill $\guarantee$ from~$\upomega_1$, the decision made by \ctrl is unsafe in this case. However, since $(\upomega_1, \upomega_1) \not\in \hpsb{\env}$, Model~\ref{alg:m2} can be proven to fulfill \guarantee with this faulty \ctrl.

So, the model is proven to fulfill \guarantee only because \env is not allowed to keep the obstacle stationary. Thus, \ctrl exploits the behavior of \env to move the obstacle so \ctrl can keep accelerating rather than stopping safely. Though \env is assumed to discard only those behaviors where it is physically impossible for \ctrl to fulfill \guarantee, the interaction between \env and \ctrl causes \env to behave in a friendly way to adapt to faulty \ctrl actions, thereby discarding \env behaviors in which $x_c$ remains constant.

\begin{problem}\label{prob:1}
    How can the \dl formula~\eqref{eq:general:dl} be guaranteed not to be valid with a  controller that exploits the environment?
\end{problem}

Observe from Fig.~\ref{fig:cheating:controller:step} that for the state~$\upomega_0$, the shaded area describes the region where the loop invariant~$\upzeta_1$ holds. The hatched area describes the states from where it is possible for \ctrl to stop before the obstacle~$x_c$, i.e., all the $xv$-points for which the \env assumption $x_c - x \geq \frac{v^2}{2\aMinN}$ in~\eqref{eq:m2:env} is true. The shaded area contains some states in the $xv$-plane that are outside of the hatched area. From these states it is not possible for \ctrl to stop before~$x_c$. Thus, control actions leading to such states should not be allowed. However, $\upzeta_1$ is not strong enough to prevent this. If $\upzeta_1$ is strengthened to allow only states contained in the hatched area then the controller is prevented from exploiting the environment. In other words, any state allowed by the loop invariant shall also be allowed by the \env assumptions, i.e., the loop invariant should imply the \env assumptions.

The assumption $x_c - x \geq \frac{v^2}{2\aMinN}$ in~\eqref{eq:m2:env} corresponds to~$P$ in the generalized Model~\ref{alg:general}. Therefore, it can be hypothesized from the above observation that the required condition to solve Problem~\ref{prob:1} can be stated as $\upzeta \to P$, where $\upzeta$ is the loop invariant and $P$ is the \env assumptions. Indeed, the condition $\upzeta \to P$ solves Problem~\ref{prob:1} for Model~\ref{alg:m2}. However, Problem~\ref{prob:1} is not specific to Model~\ref{alg:m2} and it remains unestablished whether $\upzeta \to P$ solves Problem~\ref{prob:1} for models of the general form considered in Model~\ref{alg:general}. For example, in Model~\ref{alg:m2}, the controller exploits the friendliness of \env to not keep the obstacle position~$x_c$ constant between iterations, i.e., $x_c \neq x_c^+$ for two \env actions $(x_c, x_c^+)$. Admittedly, such a behavior does not characterize friendly behavior in all models. In general, the relation between two \env actions $(e_0, e_1)$ can be any relation $R \subseteq \real \times \real$ such that $(e_0, e_1) \in R$. Note that $R$ only defines certain behaviors in the assumed operational domain. In Model~\ref{alg:m2}, the exploiting controller could be proven safe because the environment behaves \emph{friendly} by discarding some behaviors characterized by~$R$. This is illustrated in Fig.~\ref{fig:cheating:controller:obstacle} where $x_c$ cannot be kept constant as $(\upomega_1, \upomega_1) \not\in \hpsb{\env}$. 

\begin{definition}\label{def:friendly:env}
If there exists two states $\upomega_0$ and $\upomega_1$ that differ only in the assignment of the \env variable $e$, i.e., $\upomega_0(e) = e_0$ and $\upomega_1 =\upomega_0(e:=e_1)$, and such that $(e_0, e_1) \in R$ and $(\upomega_0, \upomega_1) \not\in \hpsb{\env}$, then the environment \env is \emph{friendly} w.r.t the relation $R$. Thus, \env is \emph{unfriendly} if $(e_0, e_1) \in R \to (\upomega_0, \upomega_1) \in \hpsb{\env}$ is true in all states $\upomega_0$ and $\upomega_1$ that differ only in the assignment of the \env variable $e$.
\end{definition}

The hypothesis $\upzeta \to P$ can now be generalized to include the relation $R$ to describe the existence of an unfriendly \env  as: 
\begin{nalign}\label{eq:invariant:implies:environment}
        \uprho &\equiv \forall s \ldotp \forall e \ldotp \forall e_1 \ldotp \Bigl(\upzeta(s, e) \land R(e, e_1) \to \langle \env \rangle\, (e = e_1) \Bigr)\ ,
\end{nalign}
where $\upzeta$ is parameterized to make it explicitly depend on the variables of the HP. The meaning of~$\uprho$ is that, if a state fulfills the invariant, then for every next \env action~$e_1$ characterized by $R$ there is at least one \run\ of \env in which the action~$e_1$ is chosen. 

The loop invariant $\upzeta_1 \equiv x \le x_c$ is used to prove the \dl formula~\eqref{eq:general:dl} with the definitions of Model~\ref{alg:m1} and Model~\ref{alg:m2}. Thus, it follows that $\models \upzeta_1 \to [\env;\,\sys;\,\ctrl;\,\plant]\,\upzeta_1$ holds by~\ref{it:invariant:loop}. But,~$\upzeta_1$ is not strong enough to prevent control actions that exploit friendly \env behaviors. For instance, as illustrated in Fig.~\ref{fig:cheating:controller}, the control action that leads to~$\upomega_1$ from $\upomega_0$ should not be allowed since \env must discard some behaviors from $\upomega_1$ to preserve $\upzeta_1$. These discarded behaviors include all \run s where $(\sys;\,\ctrl;\,\plant)$ do not preserve~$\upzeta_1$. Thus \ctrl \emph{exploits} \env to act friendly such that $\upzeta_1$ is preserved.

\begin{definition}\label{def:exploiting:ctrl}
A controller \ctrl \emph{exploits} a friendly environment \env w.r.t the relation~$R$ if the loop invariant~$\upzeta$ is preserved by the loop body, i.e. $\models \upgamma$, where
\begin{nalign}\label{eq:exploiting:hp:loop}
        \upgamma&\equiv\forall s \ldotp \forall e \ldotp \Bigl(\upzeta(s, e) \to [\env;\,\sys;\,\ctrl;\,\plant]\,\upzeta(s, e)\Bigr)\ ,
\end{nalign} but 
\begin{nalign}\label{eq:exploiting:hp}
        \exists s \ldotp \exists e \ldotp \exists e_0 \ldotp \Bigl(\upzeta(s, e_0) \land R(e_0, e) \land \langle \sys;\,\ctrl;\,\plant\rangle\,\lnot\,\upzeta(s, e)\Bigr)\ .
\end{nalign}
\end{definition}
Thus, \ctrl exploits \env if it makes it necessary for \env to behave friendly. In the following theorem it is shown that an exploiting controller can be prevented if the loop invariant is strong enough to ensure the existence of an unfriendly environment.  

\begin{theorem}\label{thm:cheating:controller}
    Let $s$ and $e$ be variables used in \plant and \env respectively as defined in Model~\ref{alg:general}. Let~$\upzeta(s,e)$ be a loop invariant candidate, and let~$R$ be a relation over the domain of $e$. Let $\upgamma$~\eqref{eq:exploiting:hp:loop} be the \dl formula from the inductive step~\ref{it:invariant:loop} of the loop invariant proof rule, and let $\uprho$ be as defined by~\eqref{eq:invariant:implies:environment}. If $\upgamma \land \uprho$ is valid, then the loop invariant candidate $\upzeta(s,e)$ is sufficiently strong to prevent an exploiting controller.
\end{theorem}
\begin{proof}
The following \dl formula is proved~\cite{selvaraj2022:models} in~\keyx:
    \begin{nalign}
        & \upgamma \land \uprho \to \forall s \ldotp \forall e_0 \ldotp \forall e \ldotp \Bigl(\upzeta(s, e_0) \land R(e_0, e) \to [\sys;\,\ctrl;\,\plant]\, \upzeta(s, e) \Bigr)\ .
        \label{eq:cheating:conclusion} 
    \end{nalign} 
This asserts that the loop invariant is strong enough to prevent \ctrl from exploiting \env's friendly behavior because the clause implied by $\upgamma \land \uprho$ in~\eqref{eq:cheating:conclusion} is the negation of~\eqref{eq:exploiting:hp}. 
\qed
\end{proof} 
 
In addition to solving Problem~\ref{prob:1}, \thm~\ref{thm:cheating:controller} gives hints on how the loop invariant must be constructed. In some cases, as in \fig~\ref{fig:cheating:controller} where $x_c \leq x_c^+$, it suggests that $\upzeta \equiv P$ might be a loop invariant candidate. In summary, \thm~\ref{thm:cheating:controller} is useful in two ways: \begin{enumerate*}[label=(\roman*)]
    \item By adding~$\uprho$ to a \dl formula, it is known that a proof of validity is not because \env is friendly to \ctrl, 
    \item $\uprho$ can also be a useful tool to aid in the search for a loop invariant.
\end{enumerate*}

For the specific model instance considered in this section, and probably others, changes to the model can ensure that the environment is not too friendly. However, as this paper deals with modeling errors and ascertaining that models cannot be proven valid for wrong reasons, such changes do not solve the general problem, but might nonetheless be good as best practices to avoid modeling pitfalls.

\subsection{Unchallenged Controller}\label{sec:ctrl:passenger}
The previous section dealt with modeling problems where \ctrl causes \env to exhibit friendly behaviors despite correct \env assumptions. This section discusses modeling problems due to over-constrained assumptions, whereby  \ctrl is never challenged.

Consider Model~\ref{alg:m3}, identical to Model~\ref{alg:m2}, except for \sys (\eqref{eq:m3:sys} and~\eqref{eq:m3:nomreq}). As before, $\sys$ is a black box. However, in addition to the acceleration bounds, $\sys$ also fulfills a design requirement $\nomreq$  given by~\eqref{eq:m3:nomreq}. $\nomreq$ describes that the nominal controller only requests an acceleration $a$ such that the \ego{} does not travel more than the braking distance (with $\aMinN$) from any given state in one execution of $\Tau$ duration. 
Similar to Model~\ref{alg:m2}, the requested acceleration is passed to the $\plant$ if the $\ctrl$ test $\safe$~\eqref{eq:m2:safe} is true; if not, the controller action~$C$~\eqref{eq:m2:ctrllaw} sets the maximal possible deceleration.

\begin{figure}[!t]
 \removelatexerror
  \begin{HP}[H]
    \DontPrintSemicolon
    \setstretch{0.5}
    \caption{\ctrl is unchallenged}
    \label{alg:m3}
    \SetAlgoLined
    
    \noindent
    \begin{flalign}
        \sys \triangleq &\:a\coloneqq *;\; ?\left(-\aMinN \le a \le \aMaxN \land \nomreq\right) &\label{eq:m3:sys}
    \end{flalign}
    \begin{flalign}\label{eq:m3:nomreq}
 \nomreq \triangleq &\: (v+a\Tau\ge 0) \to v\Tau + \dfrac{a\Tau^2}{2} \le \dfrac{v^2}{2\aMinN}
        \\\nonumber
             & \land (v+a\Tau < 0) \to a \le -\aMinN &
    \end{flalign}
  \end{HP}\vspace{-0.4cm}
\end{figure}

To verify that \ctrl fulfills \guarantee~\eqref{eq:m1:guarantee}, the~\dl formula~\eqref{eq:general:dl} together with the definitions in Model~\ref{alg:m1} and Model~\ref{alg:m3} must be proven valid. Though the validity can indeed be proven in~\keyx using the loop invariant $\upzeta_1 \equiv x \le x_c$, \ctrl is faulty. Strong \env and \sys assumptions might result in the invariant $\upzeta$ being true in all HP \run s irrespective of \ctrl's actions, and hence \ctrl is never verified. This manifests itself in Model~\ref{alg:m3}; \env assigns~$x_c$ such that it is possible to brake to stop before the position of the obstacle, and \sys assumes that the \ego does not travel more than the braking distance in $\Tau$ time. Therefore, \guarantee is true for all executions of $[\env;\,\sys;\,\plant]$, i.e., the model fulfills \guarantee no matter which branch of \ctrl is executed. Thus, this problem with strong \env and \sys assumptions, i.e., an over-constrained model such that \ctrl is not challenged in any HP \run, may allow a faulty controller be proven safe.

\begin{problem}\label{prob:2}
     How can the \dl formula~\eqref{eq:general:dl} be guaranteed not to be valid with an un\-challenged controller?
\end{problem}

In general, if \sys and/or \env assumptions are too strong, many relevant \run s may be discarded when the respective tests fail. A worst-case situation is when a contradiction is present in the assumption, thereby discarding all possible \run s of the HP. In that case, the \dl formula~\eqref{eq:general:dl} is vacuously true, irrespective of the correctness of \ctrl. In situations where all possible executions are discarded due to failed tests, a potential work-around is to check for such issues by proving the validity of $\assume \to \langle \env;\,\sys;\,\ctrl;\,\plant \rangle\,(\guarantee)$ to verify that there exists at least one \run\ of the hybrid program that fulfills  \guarantee. However, that work-around is not helpful to discover models susceptible to Problem~\ref{prob:2} because it is possible to prove that there is at least one \run\ of
 $(\env;\,\sys;\,\ctrl;\,\plant)$ for which $\guarantee$ is true even in over-constrained systems as seen in the HP with definitions of Model~\ref{alg:m1} and Model~\ref{alg:m3}.   

Observe that if $\ctrl$ is removed from the \dl formula~\eqref{eq:general:dl} and the formula is still valid, then $\ctrl$ is not verified. Equivalently, if the invariant is preserved when \ctrl is removed from the \dl formula, i.e., $\upchi \equiv \forall s \ldotp \forall e \ldotp \forall a \ldotp \upzeta \to [\env;\,\sys;\,\plant]\,\upzeta$ is valid, then $\ctrl$ is not verified. So the negation, i.e.,
\begin{nalign}\label{eq:environment:invalidates:invariant} 
    \lnot \upchi \equiv \exists s \ldotp \exists e \ldotp \exists a \ldotp \upzeta \land \langle \env;\,\sys;\,\plant\rangle \lnot \upzeta\ .
\end{nalign} 
can be proved to ascertain the absence of Problem~\ref{prob:2} in the proof of~\eqref{eq:general:dl}. 

\begin{definition}\label{def:challenging:ctrl}
For hybrid systems described by Model~\ref{alg:general} where the loop body is defined by (\env;\,\sys;\,\ctrl;\,\plant), 
\ctrl is \emph{challenged} w.r.t. $\env$, $\sys$, $\plant$, and the loop invariant~$\upzeta$ if $\upzeta \land \langle \env;\,\sys;\,\plant\rangle \, \lnot \upzeta$ is true in some state.
\end{definition}

However, proving $\lnot \upchi$~\eqref{eq:environment:invalidates:invariant} might not be beneficial in practice. While failed attempts to prove $\lnot\upchi$ might illuminate modeling errors, the presence of $\env$, $\sys$, $\plant$, and their interaction might complicate both the proof attempts and the identification of problematic fragments of the HP, especially for large and complicated models. 

Note that if there exists one \run\ of (\env;\,\sys) that does not preserve the invariant $\upzeta$, then $\ctrl$ must choose a safe control action such that the hybrid system can be controlled to remain within the invariant states, i.e, $\dlsb{\upzeta}$. However, this is not sufficient to conclude that the controller is verified to be safe since it could be the case that for all such invariant violating \run s, the \plant forces the hybrid system back into the invariant states. Therefore, it is necessary that not all \run s of the uncontrolled \plant  reestablish the invariant. So, if (\env;\,\sys) does not preserve the invariant, \plant does not reestablish the invariant, then \ctrl is indeed verified to be safe as shown in Theorem~\ref{thm:passenger:controller}.

\begin{theorem}\label{thm:passenger:controller}
    Let $s$, $e$, and $a$ be variables used in \plant, \env, and \ctrl respectively as defined in Model~\ref{alg:general}, and let the loop invariant candidate~$\upzeta(s,e,a_1)$ be a specific instantiation of the \dl formula~$\upzeta(s,e,a)$. Let
    \begin{nalign}
        \uppsi &\equiv \exists s \ldotp \exists e \ldotp \exists a_1 \ldotp \Bigl( \upzeta(s,e,a_1) \land \langle \env;\,\sys \rangle \Bigl(\lnot \upzeta(s,e,a)\; \land \langle \plant \rangle \lnot \upzeta(s,e,a_1)\Bigr)\Bigr)\ . \label{eq:passenger:sys:plant:breaks:inv}
    \end{nalign}
    Then, if  $\uppsi$ is valid, \ctrl is challenged in some \run s of $[\env;\,\sys;\,\ctrl;\,\plant]$.
\end{theorem}
\begin{proof}
The following \dl formula is proved~\cite{selvaraj2022:models} in~\keyx:
\begin{nalign}
    & \uppsi \to \exists s \ldotp \exists e \ldotp \exists a_1 \ldotp \upzeta(s,e,a_1) \land \langle \env;\,\sys;\,\plant \rangle \lnot \upzeta(s,e,a_1)\ . \label{eq:passenger:conclusion} 
\end{nalign}

The \dl formula~$\uppsi$~\eqref{eq:passenger:sys:plant:breaks:inv} states that there exists at least one \run\ of (\env;\,\sys) where the invariant is not preserved, and \plant does not always reestablish the invariant. The implied clause~\eqref{eq:passenger:conclusion} asserts that $\ctrl$ is challenged by Definition~\ref{def:challenging:ctrl}.
\qed\end{proof}

By the conjunction of $\uppsi$~\eqref{eq:passenger:sys:plant:breaks:inv} to a \dl formula of the form~\eqref{eq:general:dl}, Theorem~\ref{thm:passenger:controller} can be used to identify Problem~\ref{prob:2} and also the problematic fragments in all models of the form of Model~\ref{alg:general}. Furthermore, in HPs of the form $(\env;\,\ctrl;\,\plant)^*$, with no distinction between \env and \sys, Theorem~\ref{thm:passenger:controller} can still be used to determine whether the \env assumption is over-constrained. In addition, $\uppsi$ provides insights to aid in the search of a loop invariant and its dependency on the HP variables. 

\section{Results}\label{sec:results}

This section shows how Theorem~\ref{thm:cheating:controller} and Theorem~\ref{thm:passenger:controller} are used to \begin{enumerate*}[label=(\roman*)]
    \item identify that Model~\ref{alg:m2} and Model~\ref{alg:m3} are deceptive for the verification of \ctrl,   
    \item aid in the identification of a candidate loop invariant, and
    \item increase confidence in the fidelity of Model~\ref{alg:m4} where the  errors are corrected.
\end{enumerate*} 
The HPs and the \keyx proofs are available from~\cite{selvaraj2022:models}.

The \dl formula~\eqref{eq:general:dl} with the definitions in Model~\ref{alg:m1} and Model~\ref{alg:m2}, denoted as~$\uptheta$, is proved in~\keyx with the loop invariant $\upzeta_1 \equiv x \le x_c$. Therefore it follows from~\ref{it:invariant:loop} that $\models \upgamma$, where $ \upgamma \equiv \upzeta_1 \to [\env;\,\sys;\,\ctrl;\,\plant]\,\upzeta_1$.
By \thm~\ref{thm:cheating:controller}, $\uprho$ must hold for Model~\ref{alg:m2} to conclude the absence of Problem~\ref{prob:1}. 
The formula
\begin{align}\label{eq:thm:one:result:weakinv}
    \lnot \rhoone \equiv \exists x\ldotp \exists v \ldotp \exists x_c \ldotp \exists x_{c}^+ \ldotp \lnot \Bigl(x \le x_c\, &  \land\, x_c \le x_{c}^+ \to  \nonumber\\
    & \big\langle x_c\coloneqq *;\;?\,(v^2 \le 2\aMinN(x_c-x))\big\rangle\,(x_c = x_{c}^+)\Bigr),
\end{align}
expressed from~\eqref{eq:invariant:implies:environment} for Model~\ref{alg:m2} with $\upzeta(x,v,x_c) \equiv \upzeta_1$ and $R(x_c,x_{c}^+) \equiv x_c \le x_{c}^+$ is proven valid in~\keyx, thereby confirming that Model~\ref{alg:m2} is susceptible to Problem~\ref{prob:1}.  

As~$\models\lnot\rhoone$, it follows that a stronger loop invariant is needed to not verify an exploiting \ctrl. A possible candidate is the \env assumptions themselves, so let $\upzeta_2 \equiv v^2 \le 2\aMinN(x_c - x)$. For this choice of loop invariant, $\rhotwo$ is valid with $\upzeta(x,v,x_c) \equiv \upzeta_2$ and $R(x_c,x_{c}^+) \equiv x_c \le x_{c}^+$. However, $\upgamma$ cannot be proven with~$\upzeta_2$ since the \ctrl actions do not maintain~$\upzeta_2$, as already illustrated in \fig~\ref{fig:cheating:controller}. Hence, the exploiting \ctrl cannot be proven to fulfill \guarantee. These results are summarized in the first three rows of Table~\ref{tab:results}.

\begin{table}[tbp]
    \caption{Summary of validity results for incorrect and correct models}
    \centering
    \begin{tabular}{@{}lllll@{}}
        \toprule
        Model & Loop & Conjuncts & Valid & Reason \\
         & invariant & & & \\
        \midrule
        \ref{alg:m2} & $\upzeta_1$ & - & Yes & Exploiting controller\\
        \ref{alg:m2} & $\upzeta_1$ & $\rhoone$ & No & Invariant not strong enough \\
        \ref{alg:m2} & $\upzeta_2$ & $\rhotwo$ & No & Controller does not fulfill requirement \\
        \ref{alg:m3} & $\upzeta_1$ & - & Yes & Unchallenged controller \\
        \ref{alg:m3} & $\upzeta_1$ & $\lnot \chione$& No & Invariant preserved without controller\\ 
        \ref{alg:m4} & $\upzeta_1$ & - & Yes & \\
        \ref{alg:m4} & $\upzeta_1$ & $\rhoone$ $\land \lnot \chione$ & No & Invariant not strong enough \\
        \ref{alg:m4} & $\upzeta_2$ & $\rhotwo$ $\land \lnot \chitwo$ & Yes & \\
        \bottomrule
    \end{tabular}
    \label{tab:results}
\end{table}

The next two rows of Table~\ref{tab:results} summarize the results of the \dl formula~\eqref{eq:general:dl} with the definitions in Model~\ref{alg:m1} and Model~\ref{alg:m3} which is proved using the loop invariant $\upzeta_1$. Therefore it follows from~\ref{it:invariant:loop} that $\models \upgamma$. By Theorem~\ref{thm:passenger:controller}, $\models \lnot  \upchi$~\eqref{eq:environment:invalidates:invariant} must hold to ensure that \ctrl is indeed verified safe. However the~\dl formula $\chione$~\eqref{eq:thm:two:result:weakinv} with $\upzeta_1$ and \env, \sys, \plant defined by~\eqref{eq:m2:env},~\eqref{eq:m3:sys}, and~\eqref{eq:m1:plant}, respectively, is proven in~\keyx and thus, it follows that Model~\ref{alg:m3} is vulnerable to Problem~\ref{prob:2}.
\begin{nalign}\label{eq:thm:two:result:weakinv}
    \chione \equiv (x \le x_c) \to [\env;\,\sys;\,\plant]\,(x \le x_c)
\end{nalign}

\begin{figure}[!t]
 \removelatexerror
  \begin{HP}[H]
    \DontPrintSemicolon
    \setstretch{0.5}
    \caption{Correct \env, \sys, and \ctrl}
    \label{alg:m4}
    \SetAlgoLined
    \noindent
    \begin{flalign}
        \env \triangleq &\:x_c\coloneqq*;\,?\,\left(x_c-x \ge \frac{v^2}{2\aMinN} \right) & \label{eq:m4:env}
    \end{flalign}
    \begin{flalign}
        \sys \triangleq &\:a_n\coloneqq *;\; ?\left(-\aMinN \le a \le \aMaxN\right) &\label{eq:m4:sys}
    \end{flalign}
    \begin{flalign}\label{eq:m4:ctrl}
             \ctrl \triangleq \opif \, \lnot\safe(x,v,x_c,a) \, \opthen\, a\coloneqq *;\,?\,C(x,v,x_c,a)\, \opfi &&
    \end{flalign} 
    \begin{flalign}\label{eq:m4:safe}
        \safe \triangleq x_c - x \ge v\Tau + \frac{\aMaxN\,\Tau^2}{2} + \frac{\left(v+\aMaxN\,\Tau\right)^2}{2\aMinN} &&
    \end{flalign}
    \begin{flalign}\label{eq:m4:ctrl:safeaction}
        C(x,v,x_c,a) \triangleq a = -\aMinS &&
    \end{flalign}
  \end{HP}\vspace{-0.4cm}
\end{figure}

The last three rows of Table~\ref{tab:results} summarize the results of the \dl formula~\eqref{eq:general:dl} with the definitions in Model~\ref{alg:m1} and Model~\ref{alg:m4}, where all three parts conjuncted together is denoted by~$\upkappa$. Based on the insights about Model~\ref{alg:m2} and Model~\ref{alg:m3} from Table~\ref{tab:results}, Model~\ref{alg:m4} rectifies Problem~\ref{prob:1} and Problem~\ref{prob:2}. Similar to the previous models, the \env assumption~\eqref{eq:m4:env} assigns $x_c$ such that it is possible to brake to stop before the obstacle and \sys ~\eqref{eq:m4:sys} is a black box. Unlike the previous models, the \ctrl test \safe in~\eqref{eq:m4:safe} not only checks whether the worst-case acceleration is safe in the current execution but also checks whether, in doing so \guarantee is fulfilled in the next loop execution. 

The~\dl formula $\upkappa$ is proved in~\keyx using the loop invariant $\upzeta_1 \equiv x \le x_c$. Since $R(x_c,x_{c}^{+}) \equiv x_c \le x_{c}^{+}$ is also applicable for Model~\ref{alg:m4}, it follows from $\models \lnot \rhoone$~\eqref{eq:thm:one:result:weakinv} that $\upzeta_1$ is not sufficiently strong to solve Problem~\ref{prob:1}. The stronger invariant candidate $\upzeta_2 \equiv v^2 \le 2\aMinN(x_c - x)$ is used to prove $\upkappa$ and since $\models \rhotwo$, it is concluded that~$\upzeta_2$ is sufficiently strong to solve Problem~\ref{prob:1} for~Model~\ref{alg:m4}.

Finally, to confirm that Model~\ref{alg:m4} is not susceptible to Problem~\ref{prob:2}, $\uppsi$ from \thm~\ref{thm:passenger:controller} must hold. The \dl formula $\psitwo$~\eqref{eq:thm:two:result:stronginv} is proven in~\keyx:

\begin{nalign}\label{eq:thm:two:result:stronginv}
    \psitwo \equiv \exists x \ldotp \exists v \ldotp   \exists x_c \ldotp \Bigl( \upzeta&(x,v,x_c,\aMinN)\, \land\\ 
    &\langle \env;\,\sys \rangle\, \Bigl(\lnot \upzeta(x,v,x_c,a)
     \land \langle \plant \rangle \lnot \upzeta(x,v,x_c,\aMinN)\Bigr)\Bigr)\, ,
\end{nalign}
where \env, \sys and \plant are as defined in~\eqref{eq:m4:env}, \eqref{eq:m4:sys} and~\eqref{eq:m1:plant} respectively, and the loop invariant $\upzeta_2 \equiv \upzeta(x,v,x_c,\aMinN)$ is a specific instantiation of the \dl formula  $\upzeta(x,v,x_c,a)$ given by:
\begin{nalign}\label{eq:zetaiterate}\nonumber
    \upzeta(x,v,x_c,a) \equiv\; & (v+a\Tau\ge0) \to (v+a\Tau)^2 \le 2\aMinN \left(x_c - x - v\Tau - \frac{a\Tau^2}{2}\right)\,\land \\ \nonumber
                                    & (v+aT < 0) \to v^2 \le 2\aMinN (x_c - x)\ .
\end{nalign}
With this result, it holds that $\models \psitwo$, and therefore it follows from Theorem~\ref{thm:passenger:controller} that $\models \lnot \chitwo$ for the choice of $\upzeta_2$. Thus, it entails that Model~\ref{alg:m4} is bereft of Problem~\ref{prob:1} and Problem~\ref{prob:2}, as summarized in the last row of Table~\ref{tab:results}.

\section{Related Work}

The models considered in this paper are similar to the models used to verify the European Train Controller System~(ETCS)~\cite{platzer2009:etcs}. Though not explicitly stated, the modeling pitfalls are avoided for the ETCS models by the use of an \emph{iterative refinement process} that determines a loop invariant based on a controllability constraint. The process is used to design a correct controller rather than to verify one.

An alternative to guarantee CPS correctness is \emph{runtime validation}~\cite{mitsch2016:modelplex}, where runtime monitors are added to the physical implementation, monitoring whether the system deviates from its model. If it does, correctness is no longer guaranteed, and safe fallbacks are activated. However, for Model~\ref{alg:m2}, the safe fallback would be activated too late since \ctrl had already taken an unsafe action when the violation of the \env assumptions are detected. Furthermore, the safe fallbacks might cause spurious braking for Model~\ref{alg:m3}.

The issue in Model~\ref{alg:m2} is not unique to \dl; the issue manifests itself similarly in reactive synthesis~\cite{bloem2014:assumptions,majumdar2019ecoGR1}. The cause of the issue, in both paradigms, stems from the logical implication from the \env assumptions to the \ctrl actions and requirements. Instead of taking actions to fulfill the consequent, an exploiting \ctrl can invalidate the premise to fulfill the implication. However, Bloem et al.~\cite{bloem2014:assumptions} conclude that none of the existing approaches completely solve the problem and emphasize the need for further research.

Theorems~\ref{thm:cheating:controller} and~\ref{thm:passenger:controller} put conditions on individual components, but these conditions, in the form of the loop invariant, stem from the same global requirement. M\"{u}ller et al.~\cite{muller2018contract:composition} take the other approach and start with separate requirements for each of the components to support the global requirement. The goal of the decomposition is to ease the modeling and verification effort, and not directly to validate the model. However, these methods would likely be beneficial in tandem.

The contributions of this paper give additional constraints, apart from the three implications of the \looprule, that can aid the construction of invariants. This might be useful in automatic invariant inference, which is a field of active research where loop invariants are synthesized. Furia and Meyer~\cite{furia2010infering:invariants} note that the automatic synthesis of invariants based on the implementation (or the model) might be self-fulfilling, and go on to argue that the postconditions and the global requirements must be considered in the invariant synthesis. This paper, however, suggests that, for certain models, it might not be sufficient to consider only the postconditions in the invariant synthesis.

\section{Conclusion}

Modeling errors present a risk of unsound conclusions from provably safe erroneous models, if used in the safety argument of safety-critical systems. This paper formulates and proves conditions in \thm~\ref{thm:cheating:controller} and \thm~\ref{thm:passenger:controller} that, when fulfilled, help identify and avoid two kinds of modeling errors that may result in a faulty controller being proven safe. Furthermore, the formulated conditions aid in finding a loop invariant which is typically necessary to verify the safety of hybrid systems.

Using a running example of an automated driving controller, the problematic cases are shown to exist in practical CPS designs. The formulated conditions are then applied to the erroneous models to show that the errors are captured. Finally, the errors are rectified to obtain a correct model, which is then proved using a loop invariant that satisfies the formulated conditions, thus ensuring absence of the two modeling errors discussed in this paper.   

A natural extension of this work will be to investigate also other kinds of modeling errors that might arise in the verification of complex CPS designs. Moreover, it would also be beneficial to investigate the connection between loop invariants and differential invariants, which are used to prove properties about hybrid systems with differential equations without their closed-form solutions. 

\newpage

%
%
\bibliographystyle{splncs04}
\bibliography{mybib}

\end{document}

%% file: preamble.tex
\usepackage[T1]{fontenc}
\usepackage{graphicx}
\usepackage{upgreek}
\usepackage{graphicx}
\usepackage{amsfonts}
\usepackage{amsmath}
\usepackage{amssymb}
\usepackage{xspace}
\usepackage{mathtools}
\usepackage{tikz}
\usetikzlibrary{positioning,arrows.meta,quotes,chains,shapes,fit,calc,patterns}
\usepackage{subcaption}

\usepackage{dirtytalk}
\usepackage{tabularx}
\usepackage{booktabs}
\usepackage{txfonts} 
\usepackage[mathscr]{euscript} 
\usepackage[inline]{enumitem}

\usepackage{stmaryrd}
\usepackage{bussproofs} 

\usepackage{tikzsymbols}
\usepackage[ruled,norelsize]{algorithm2e} 

\newenvironment{nalign}{
    \begin{equation}
    \begin{aligned}
}{
    \end{aligned}
    \end{equation}
    \ignorespacesafterend
}

\makeatletter
\newenvironment{HP}[1][H]{%
   \begin{algorithm}[#1]%
  }{\end{algorithm}}
\makeatother  

\makeatletter
\newcommand{\removelatexerror}{\let\@latex@error\@gobble}
\makeatother

\usepackage{setspace} 
\usepackage{siunitx}

\newcommand{\dl}{\text{\upshape\textsf{d{\kern-0.03em}L}}\xspace}
\newcommand{\keyx}{KeYmaera\,X\xspace}

\newcommand{\dlold}{\text{\upshape\textsf{d{\kern-0.05em}$\mathcal{L}$}}\xspace}

\newcommand{\stateset}{\ensuremath{\mathscr{S}}\xspace}
\newcommand{\varset}{\ensuremath{\mathscr{V}}\xspace}
\newcommand{\real}{\ensuremath{\mathbb{R}}\xspace}
\newcommand{\naturalnum}{\ensuremath{\mathbb{N}}\xspace}
\newcommand{\rational}{\ensuremath{\mathbb{Q}}\xspace}

\newcommand{\lsb}{\ensuremath{\llbracket}} 
\newcommand{\rsb}{\ensuremath{\rrbracket}}
\newcommand{\sbb}[1]{\lsb #1 \rsb}

\newcommand{\lsp}{\ensuremath{\llparenthesis}}
\newcommand{\rsp}{\ensuremath{\rrparenthesis}}
\newcommand{\spp}[1]{\lsp\mkern1mu #1 \mkern1mu\rsp}

\newcommand{\lspn}{\ensuremath{(}}
\newcommand{\rspn}{\ensuremath{)}}
\newcommand{\sppn}[1]{\lspn #1 \rspn}

\newcommand{\dlsb}[1]{\spp{#1}}
\newcommand{\hpsb}[1]{\sbb{#1}}
\newcommand{\ssb}[1]{\sppn{#1}}

\newcommand{\Tau}{T} 

\newcommand{\aMinN}{\ensuremath{a_n^{min}}\xspace}
\newcommand{\aMaxN}{\ensuremath{a_n^{max}}\xspace}
\newcommand{\aMinS}{\ensuremath{a_s^{min}}\xspace}

\newcommand{\opif}{\ensuremath{\operatorname{if}\xspace}}
\newcommand{\opthen}{\ensuremath{\operatorname{then}\xspace}}

\newcommand{\opfi}{\ensuremath{\operatorname{fi}\xspace}}

\newcommand{\safe}{\ensuremath{\mathit{ok}}\xspace} 
\newcommand{\assume}{\ensuremath{\mathit{init}}\xspace}
\newcommand{\guarantee}{\ensuremath{\mathit{guarantee}}\xspace}

\newcommand{\ego}{ego-vehicle\xspace}

\newcommand{\sys}{\ensuremath{\mathit{aux}}\xspace}
\newcommand{\env}{\ensuremath{\mathit{env}}\xspace}
\newcommand{\ctrl}{\ensuremath{\mathit{ctrl}}\xspace}
\newcommand{\plant}{\ensuremath{\mathit{plant}}\xspace}

\newcommand{\nomreq}{\emph{req}\xspace}

\newcommand{\run}{execution}

\newcommand{\looprule}{\texttt{loop} rule}

\newcommand{\thm}{Theorem}

\newcommand{\fig}{\figurename}

\newcommand{\rhoone}{\uprho^{}_{1}}
\newcommand{\rhotwo}{\uprho^{}_{2}}

\newcommand{\psitwo}{\uppsi^{}_{2}}
\newcommand{\chione}{\upchi^{}_{1}}
\newcommand{\chitwo}{\upchi^{}_{2}}

\usepackage[hidelinks]{hyperref}

\usepackage[numbers]{natbib} 
\makeatletter 
\renewcommand\bibsection%
{ \section*{\refname 
    \@mkboth{\MakeUppercase{\refname}}{\MakeUppercase{\refname}}} 
} 
\makeatother